\documentclass[runningheads]{llncs}
\usepackage{graphicx}
\usepackage{subfigure}
\usepackage{booktabs}
\usepackage{multirow}
\usepackage[ruled,linesnumbered]{algorithm2e}
\usepackage{algpseudocode}
\usepackage{amsmath}
\usepackage{amsfonts,amssymb}
\usepackage{minibox}

\begin{document}
	\title{BPPChecker: An SMT-based Model Checker on Basic Parallel Processes(Full Version)}
	%
	%
	\author{Ying Zhao\inst{1} \and
		Jinhao Tan\inst{2} \and
		Guoqiang Li\inst{1}
	}
	\authorrunning{Y. Zhao et al.}
	%
	\titlerunning{BPPChecker: An SMT-based Model Checker on Basic Parallel Processes}
	\institute{Shanghai Jiao Tong University, Shanghai, China \\
		\email{\{zhaoying98, li.g\}@sjtu.edu.cn} \\
		\and
		The University of Hong Kong, Hong Kong, China \\
		\email{tjh19962013@163.com}}
	\maketitle  
	\begin{abstract}
		Program verification on concurrent programs is a big challenge due to general undecidable results. Petri nets and its extensions are used in most works. However, existing verifiers based on Petri nets are difficult to be complete and efficient. Basic Parallel Process (BPP), as a subclass of Petri nets, can be used as a model for describing and verifying concurrent programs with lower complexity.
		We propose and implement BPPChecker, the first model checker for verifying a subclass of CTL on BPP. We propose constraint-based algorithms for the problem of model checking on BPPs and handle formulas by SMT solver Z3. For EF operator, we reduce the model checking of EF-formulas to the satisfiability problem of existential Presburger formula. For EG operator, we provide a $k$-step bounded semantics and reduce the model checking of EG-formulas to the satisfiability problem of linear integer arithmetic. Besides, we give Actor Communicating System (ACS) the over-approximation BPP-based semantics and evaluate BPPChecker on ACSs generated from real Erlang programs. Experimental results show that BPPChecker performs more efficiently than the existing tools for a series of branching-time property verification problems of Erlang programs.

		
	\end{abstract}

	\section{Introduction}
	Program verification on concurrent programs is a big challenge due to general undecidable results~\cite{ramalingam2000,107}. Most works are on Petri nets~\cite{petri1966communication} and its extensions like multiset pushdown system~\cite{21}, Petri Nets with Unordered Data~\cite{22} and Nets with Nested Colored Tokens (NNCT)~\cite{24}. It was proved that reachability on Petri nets has an ACKERMANN upper bound~\cite{25} and a Tower-hard lower bound~\cite{26}. Czerwińsk further improved the lower bound by increasing the height of tower from linear to exponential and proved that without restriction on dimension, the problem needs a tower of exponentials of time or space, of height exponential in input size~\cite{czerwinski2021improved}. For coverability and boundedness on Petri nets, the complexity is EXPSPACE-complete~\cite{28,9,atig2009context}. So existing automatic tools such as BFC~\cite{10} and Petrinizer~\cite{19} cannot perform well for large-scale program verification and the tools are difficult to be complete and efficient.
	
	Basic Parallel Process (BPP) is an important subclass of Petri net~\cite{13} which still holds some concurrent properties. On BPP, coverability and reachability are NP-complete which can be efficiently handled by SAT/SMT solvers~\cite{29,121,SMT}. Although liveness (\textbf{EG}) is still undecidable on BPP, the hierarchical structure of BPP allows us to perform a bounded model checking on liveness. By using BPP, we can greatly reduce the complexity of model checking in theory, and implement efficient algorithms and practical tools to verify asynchronous communicating programs~\cite{18}. 

	We propose and implement BPPChecker, the first model checker for verifying a subclass of CTL (CTL$_\lnot U$) on BPP. For the NP-complete reachability (EF-formula), inspired by the reduction algorithm proposed by Verma~\cite{13}, we reduce it to the satisfiability of existential Presburger formulas. For the undecidable liveness (EG-formula), we propose a $k$-step bounded semantics of EG-formula on BPP and reduce the bounded liveness problem to satisfiability of linear integer arithmetic formulas. The linear integer arithmetic formulas generated by our algorithms are handled by SMT solver Z3. 
	
	We give Actor Communicating System (ACS) the over-approximation BPP-based semantics and evaluate BPPChecker on ACSs generated from real Erlang programs~\cite{erlang}. ACS is a sound model of Erlang program that generated by the Erlang verifier Soter~\cite{1}. By means of Soter, we can easily transfer Erlang programs to ACS and then perform model checking ACS with support of ACS2BPP module in our tool. Experimental results show that BPPChecker has advantages in speed and the number of constraints generated and it has more efficient results than existing tools in a series of reachability and bounded liveness property verification.
	
	\subsubsection{Contributions.} To summarize, this paper makes the following contributions:
	\begin{itemize}
		\item A BPPChecker that supports model checking EG-formulas and EF-formulas on BPP;
		\item The support of ACS2BPP module in BPPChecker that can help us perform model checking on Erlang programs.
	\end{itemize}
	
	The remainder of this paper is structured as follows. Section 2 provides the necessary preliminaries. Section 3 describes our model checking algorithms of EF-formulas and EG-formulas on BPP. Section 4 proposes our reduction from ACS to BPP. Section 5 gives an experimental evaluation of BPPChecker and an analysis of experimental results. Section 6 describes related work. Section 7 concludes the paper and discusses our future work.
	
	\section{Preliminaries}
	
	\subsection{Basic Parallel Process}
	Basic Parallel Process(BPP) is a model for asynchronous concurrent systems, which can be regarded as a subclass of Petri nets~\cite{christensen1993decidability}. It models the processes of a concurrent system as a symbol and models a state as a concurrent combination of multiple symbols. A symbol can produce more symbols through transition, so a concurrent process usually produces an infinite state system. 
	
	A BPP expression contains action prefixes, choice and merge operations, respectively representing process transition, choice of process transition and combination of concurrent processes. The semantics of transition in BPP are asynchronous. According to~\cite{7}, BPP is regarded as a commutative context-free grammar, which will be adopted in this paper. 
	
	Let $ Var=\{X,Y,Z,...\} $ be the set of variables and $ Act = \{a,b,c,...\} $ be the set of actions.
	
	\begin{definition} A BPP is a 2-tuple $ (V, \Delta) $, where $ V \in Var $ is a finite set of symbols and $ \Delta $ is a finite set of rules. Rules of $ \Delta $ is written as $ X \rightarrow \alpha $, where $ X \in V, a \in Act, \alpha \in V^{\oplus} $. $ V^{\oplus} $ represents free commutative monoid generated by $ V $. A BPP$(V, \Delta)$ determines a labelled transition system$ (V^{\oplus}, Act, \rightarrow, \alpha) $, where $ V^{\oplus} $ is the state space, $ \rightarrow $ is the transition relation generated by the following rule:
		\begin{equation}
			\frac{X \stackrel{a}{\longrightarrow} \alpha \in \Delta}
			{\beta X \gamma \stackrel{a}{\longrightarrow} \beta\alpha\gamma},
		\end{equation}
		where the BPP expression $ \alpha, \beta, \gamma \in V^{\oplus}$.
	\end{definition}
	
	We denote a reflexive transitive closure of single-step transition relation $ \{\stackrel{a}{\longrightarrow}\}_{a \in Act}$ by  $ \rightarrow^\star $. For instance, $ \alpha \rightarrow^\star \beta $ denotes that state $ \beta $ is reachable from state $ \alpha $ after several transitions. 
	
	A BPP expression has modulo commutativity, which means $XYZ$, $ YXZ $, $ ZXY $ are treated as the same element. So a BPP expression is intuitively a concurrent combination of symbols, each of which can transfer independently according to its own rules. 
	
	\begin{example}
		\label{bpp}
		Given a BPP$ (V, \Delta) $, where symbol set $ V = \{X_1, X_2, X_3\} $ and rule set $ \Delta $ consists of rules 
		$ r_1: X_1 \stackrel{a}{\longrightarrow} X_2X_3$,
		$ r_2: X_2 \stackrel{a}{\longrightarrow} X_1X_2$,
		$ r_3: X_3 \stackrel{a}{\longrightarrow} X_1$. Start from the BPP expression $ X_1 $, the state transition is shown in Figure~\ref{fig:bpp}.
	\end{example}

	\begin{figure}[htbp]
		\centering
		\includegraphics[width=0.5\textwidth]{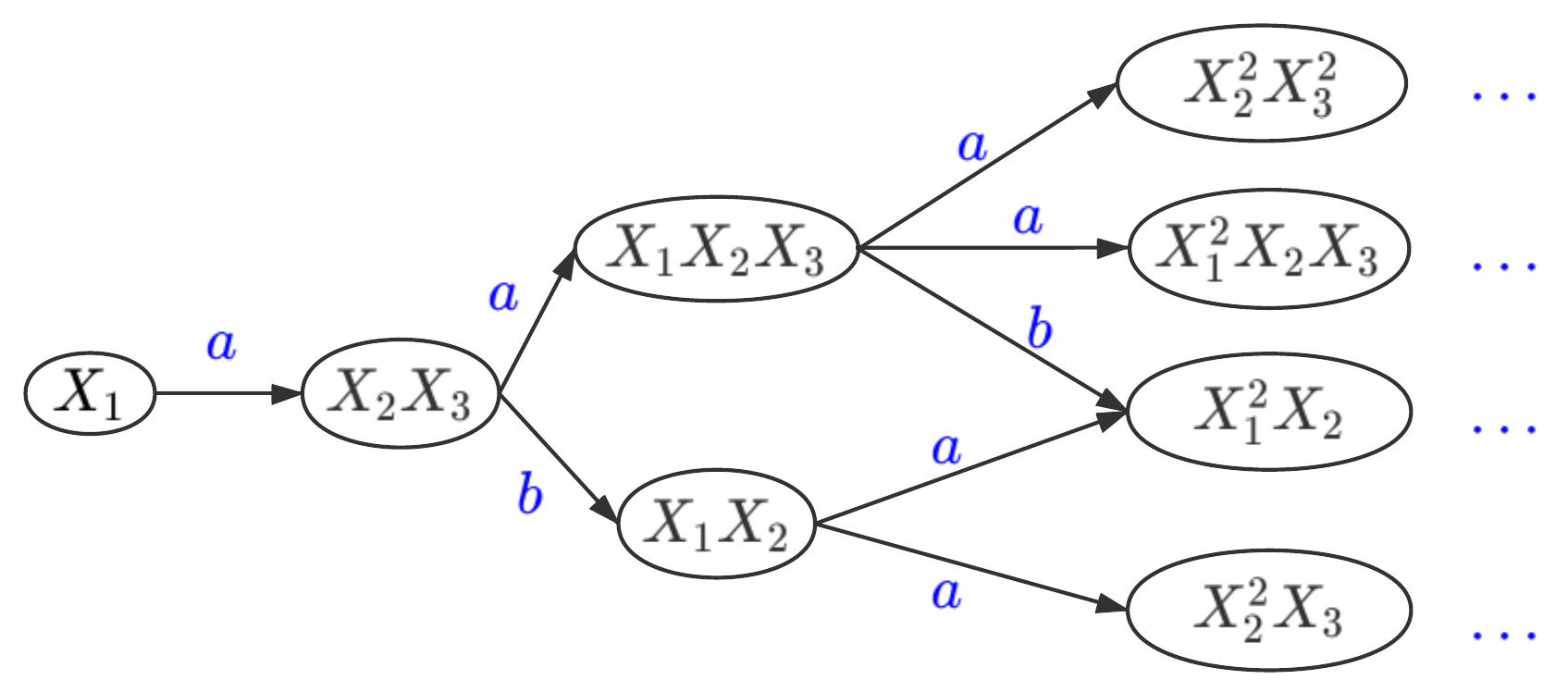}
		\caption{The transition of BPP in Example~\ref{bpp}}		\label{fig:bpp}
	\end{figure}

	\subsection{CTL}
Computation Tree Logic (CTL) is a branching-time logic, which means the model of time is a tree-like structure, in the future there're different paths. A CTL formula uses logical operators and temporal operators, the logical operators are the usual ones: $\wedge,\vee,\lnot$; the temporal operators include quantifiers over paths $A,\ E$ which means all paths and exist one path and path-specific quantifiers:$X,G,U,F$. $X$ means 'Next', $G$ means 'Globally', $U$ means 'Until' and $F$ means 'Finally'. One minimal set of operators is: \{true, $ \vee, \lnot, \textbf{EG}, \textbf{EU}, \textbf{EX} $\}.
In this paper we focus on the fragment of CTL called $\mathrm{CTL}_{\lnot U}$ with the following syntax

\begin{equation}
	\varphi ::= \lambda | \lnot\varphi  | \varphi \wedge \varphi |\varphi \vee \varphi| \varphi \Rightarrow \varphi|\textbf{E}	\left \langle a \right \rangle \varphi | \textbf{EG} \varphi | \textbf{EF} \varphi | \textbf{EX} \varphi,
\end{equation}
where $ a\in Act $. $ \lambda $ is an atomic fomrula of the form as $ \textbf{am}^T \ge \textbf{b} $, where $ \textbf{m} = (X_1,...,X_n) $, $a=(a_1,...,a_n) \in \mathbb{N}^n, b \in \mathbb{N} $.
\begin{example} 
	The the duality \textbf{AF} operator can be defined as: $ \textbf{AF}\varphi \stackrel{\mathrm{def}}{=} \lnot \textbf{EG}(\lnot \varphi) $. 
	The formula $ \textbf{AF}(X+Y \ge 3 \rightarrow \textbf{E}	\left \langle a \right \rangle(Z \ge 1)) $ denotes that " If the sum of $ X $ and $ Y $ in a state is at least 3, then this state exists a subsequent state passing action $ a $ whose number of the process $ Z $ is greater than or equal to 1" will be satisfied at some point in the future. 
\end{example}

If a $\mathrm{CTL}_{\lnot U}$ formula contains no \textbf{EF} operator we call it an \textbf{EG}-formula; if it contains no \textbf{EG} operator we call it an \textbf{EF}-formula. In terms of the complexity of model checking problem for \textbf{EG}-formulas(i.e. liveness problem) on BPPs, Esparza~\cite{29} proved the following theorem:
\begin{theorem}
	\label{theorem_EG}
	The model checking problem for \textbf{EG}-formulas on BPPs is undecidable.
\end{theorem}

\textbf{EF}-formulas can define the reachability properties~\cite{berard2001reachability}. For the complexity of model checking \textbf{EF} (i.e. reachability problem) on BPPs, Esparza~\cite{28} proved the following theorem:

\begin{theorem}
	\label{theorem_EF}
	The model checking problem for \textbf{EF}-formulas on BPPs is NP-complete.
\end{theorem}
	
	\subsection{Actor Communicating System}
	Actor Communicating System (ACS) is a sound model for asynchronously communicating programs proposed by Osualdo~\cite{1}. ACS contains behaviors of creating processes, sending messages and receiving messages. The formal definition of ACS is defined as follows.
	
	\begin{definition}  An ACS is a 4-tuple $(Q,P,M,R)$, where $ Q $ is a finite set of control states, $ P $ is a finite set of process states, $ M $ is a finite set of messages and $ R \subseteq Q \times Operations \times Q$ is a finite set of transition rules. For $ r \in R $, $ r $ can be written as $ q_1 \stackrel{op}{\longrightarrow} q_2$, where $ q_1, q_2 \in Q $ and $ op \in \{nop, \nu q_0$, p!m (send message $ m $ to process $ p $), $p?m$ (receive message $ m $ from process $ p $)\}. 
	\end{definition}
	We define $ \Rightarrow $ that denotes $ \bigcup_{r\in R} \stackrel{r}{\Rightarrow} $, and $ \Rightarrow^\star $ that denotes the reflexive transitive closure of $ \Rightarrow $.
	Semantically, ACS does not remember the order of messages in a mailbox, instead it uses a counter abstraction on the mailbox,
	recording the number of messages in the mailbox. ACS also use a second counter on the control state of each process class to count the process in the current control state. So we further use a VAS (short for Vector Addition System~\cite{8}) to express the semantics of ACS intuitively. 
	
	Thus The semantics of ACS can be defined as a transition system generated by VAS $ \Upsilon (I, R)$, where the set of places $ I = Q \cup (P \times M) $. We denote the vector set $ (u,v) \in \mathbb{N}^I$ as a place of $ \Upsilon $, where vector $ u=(q_1, \cdots, q_{|Q|}) $ contains counters of all states and vector $ u=((p_1, m_1), (p_1, m_2),\cdots) $ contains counters of messages in mailbox. Finally, we introduce two self-defined symbols: $ u[q_1] $ denotes the number of state $ q_1 $ and symbol $ v[(p_1, m_1)] $ denotes the count of $ m_1 $ in $ p_1 $'s mailbox.

	\section{Model Checking EF-formulas and EG-formulas}
	\label{sec3}
	This section provides our reduction of model checking \textbf{EF}-formulas and \textbf{EG}-formulas on BPP. First, we reduce the model checking problem of \textbf{EF}-formula to the satisfiability of existential Presburger formulas~\cite{stansifer1984presburger}. Second, we reduce the bounded model checking problem of \textbf{EG}-formula to the satisfiability of linear integer arithmetic (LIA) in $ k $ steps.
	
	\subsection{EF-formulas: Construction of Existential Presburger Formula}
	\label{reachsec}
	The existing idea of NP-complete model checking is to reduce the  problem to smaller computing objects and solve it by efficient SMT solvers. Verma et al. proved that the Parikh image of context-free grammar could be reduced to the satisfiability of existential Presburger formula~\cite{13}. However, there is a problem with this reduction, which makes the formula constructed too few constraints. Barner proposed a new version and corrected the problem~\cite{14}. We implements the reduction method equivalent to Barner's version logic proposed by Hague~\cite{15}.
	
	Given a BPP expression $ (V,\Delta)$ and an initial BPP expression $ \alpha $. We assume that $\alpha$ contains only one symbol, and let $\alpha = P_{init} \in V$. We introduce variables:
	\begin{itemize}
		\item $\forall P \in V $, introduce $ x_p $ representing the number of times that process $ P $ appears;
		\item $\forall r \in \Delta $, introduce $ y_r $ representing the number of times that rule $ r $ is used;
		\item According to the transition rule, define a spanning tree whose nodes are process symbols and introduce $ z_p $ representing the distance from $ P $ to $ P_{init} $.
	\end{itemize}
	Now we can construct the existential Presburger formula, which contains two parts of constraints. The first part mainly expresses that the number of times a transition rule is used must be consistent with the number of symbols appear in the BPP expression. For example, if $ r: P \rightarrow YZ $ is the only rule that generates $ Y $(i.e. $ Y \in r^{\bullet} $), then $ y_r == x_Y $. As described above, the first part consists of following constraints: 
	\begin{itemize}
		\item For $ P \in V $, introduce constraint: $ x_p \ge 0 $
		\item For $ r \in \Delta $, introduce constraint: $ y_r \ge 0 $ 
		\item Let $ r_1, \cdots, r_k $ be rules whose left symbol is $ P $, note that symbol $ \alpha(P) $ represents the number of $ P $ in $ \alpha $, introduce constraint: 
		\begin{equation}
			\alpha(P) + \sum_{r \in \Delta} y_r r^{\bullet}(P) - \sum_{i=1}^k y_{r_i} = x_p
		\end{equation}
	\end{itemize}
	
	For any reachable state $ \alpha $, the second part rely on $ z_p $ to describe the precondition that rules that generate $ \alpha $ must be used at least once.
	\begin{itemize}
		\item For $ P \in V $, introduce constraint: $ x_p = 0 \vee z_p > 0 $
		\item Let $ r_1, \cdots, r_l $ be rules whose right symbol is $ P $, $ Y_1, \cdots, Y_l $ be left symbols of corresponding rules, introduce constraint: 
		\begin{equation}
			(z_p = 0 \wedge \bigwedge_{i=1}^l y_{r_i} = 0) \vee
			\bigvee_{i=1}^l (z_p = z_{Y_i} +1 \vee y_{r_i} > 0 \vee z_{Y_i} > 0) 
		\end{equation}
	\end{itemize}
	
	\subsection{Bounded EG-formulas: Linear Integer Arithmetic}
	\label{lsec}
	According to Theorem~\ref{theorem_EG}, model checking \textbf{EG} on BPP is undecidable so we cannot efficiently perform model checking liveness on BPP under standard semantics. A labelled transition system generated by a BPP and an initial BPP expression may have infinitely long but non-cyclic paths, resulting in an increasing number of process symbols of BPP states. In reality, however, most programs can terminate within a finite number of steps, so the number of processes does not tend to be infinite and the transition of a process tend to stop within a certain number of steps. Based on these observations, we present an approach for bounded model checking of \textbf{EG}-formulas on BPP, which is combined with ideas from our previous work~\cite{3,4,27}. Our approach contain three parts: (i) Proposing the $ k $-step bounded semantics of \textbf{EG}: $ s \models_k \varphi $. (ii) Constructing the corresponding linear integer arithmetic (LIA) for the given BPP and \textbf{EG}-formulas. (iii) Using our tool to solve the satisfiability of linear integer arithmetic.
	\subsubsection{K-step Bounded Semantics.}
	We give the $k$-step bounded semantics for \textbf{EG}-formulas, limits a BPP state to satisfying properties within $k$ steps.
	\begin{definition}
		Let $s$ be a BPP expression, $ \varphi $ be an \textbf{EG}-formula, $ k \ge 0 $, then the $K$-step bounded semantics $ s \models_k \varphi $ is inductively defined as:
		\begin{itemize}
			\item $ s \models_k \textbf{am}^T \ge b $ iff $\textbf{a}s^T \ge b $,
			\item $s \models_k \lnot\varphi $ iff $s \models_k \varphi $ is false,
			\item $ s\models_k \varphi_1 \wedge \varphi_2$ iff $ s\models_k \varphi_1$ and $ s\models_k \varphi_2 $,
			\item $ s \models_k \textbf{E} \left \langle a \right \rangle \varphi $ iff $ k \ge 1 $ and there exists a state $ t $ such that $ s\stackrel{a}{\longrightarrow} t $ and $ t \models_k \varphi $,
			\item $ s \models_k \textbf{EG}\varphi $ iff there exists a path $ \pi $ such that $\pi(0) = s $ and $ \forall i \ge 0, \pi(i) \models_k \varphi$.
		\end{itemize}
	\end{definition}
	Note that for $ s \models_k \textbf{E} \left \langle a \right \rangle \varphi $, since the satisfication of $ \textbf{E} \left \langle a \right \rangle \varphi $ needs at least 1 transition so the number of steps has a lower bound 1.
	
	\subsubsection{Algorithm for LIA Construction.}
	Linear integer arithmetic is a first-order theory and the syntax of linear integer arithmetic is : 
	$ \psi ::= a_0+a_1x_1+\cdots+a_nx_n \triangleleft 0 | \lnot\psi | \psi \wedge \psi | \exists x.\psi $, where $ a_0, a_1, \cdots, a_n \in \mathbb{Z}, \triangleleft \in \{>, =\} $.
	
	According to $ k $-step bounded semantics of EG-formulas given above, we can construct the corresponding LIA. we first define some basic constraints to represent BPP's behavior and then give an algorithm to construct LIA.
	\subsubsection{Constraints.} Firstly, we introduce some basic symbols and mappings. Similar with the definition in Petri nets, for each rule $ r\in \Delta $, define $^{\bullet}r$ as symbols on the left, $ r^{\bullet} $ as BPP expressions on the right and $ r^{act} $ as the action. For example, if $ r= X_1 \stackrel{b}{\longrightarrow} X_2 $, then $^{\bullet}r = X_1, r^{\bullet} = X_2 $ and $ r^{act}=b $.
	
	Define mapping $P^{-}$ like Parikh mapping $ P $, where each element represents the number of symbols - 1. For example, $ P^{-}(X_2, X_1 X_2 X_2 X_3 X_3)=(1,1,2)$. Given $ r\in \Delta $, n-dimensional vector $ s=(s_1,..., s_n) $ and $ t=(t_1,..., t_n) $, define $T^{-}(s,t,r)$ as:
	\begin{equation}
		T^{-}(s,t,r) = \bigwedge_{i=1}^n (s_i + P^{-}(^{\bullet}r, r^{\bullet})_i = t_i)
	\end{equation} 
	
	Note that $T^{-}(s,t,r)$ reflects the change of the number of process symbols after $ s $ arrives at $ t $ through transition rule $ r $. Based on $T^{-}(s,t,r)$, we now give the definition of transition constraints $ T(s,t,r) $ (i.e. $ s \stackrel{a}{\longrightarrow} t $) that restricts each component of a legal BPP expression to be non-negative:
	\begin{equation}
	 	T(s,t,a) = \bigvee_{r\in \Delta} (r^{act} = a \wedge s(^{\bullet}r) \ge 1 \wedge T^{-}(s,t,r))
	 \end{equation} 
	
	$ s(^{\bullet}r) \ge 1 $ is necessary for the transition to be triggered, guaranteeing the requirement of BPP that occurrences of each process symbol in BPP expression cannot be negative. We define path constraints $ Path((u(0), \cdots, u(k))) $ as: 
	\begin{equation}
		Path((u(0), \cdots, u(k))) =\bigwedge_{j=1}^k [\bigvee_{r\in \Delta} (u(j-1)(^{\bullet}r)\ge 1 \wedge T^{-}(u(j-1),u(j),r))] 
	\end{equation}

$ u(j-1) $ can use the rule $ r \in \Delta$ to arrive $ u(j) $, which corresponds to $ u(j-1)(^{\bullet}r)\ge 1 $.

	The Algorithm generating corresponding LIA is presented in Algorithm~\ref{algo}. The algorithm is a recursive function, which recurses according to the structure of \textbf{EG}-formula $ \varphi $. If $ \varphi $ is an atomic formula, the corresponding constraints in the semantic definition are generated. If the outermost layer of $ \varphi $ is a logical proposition operator, it is constructed according to the semantics of negation and conjunction. For $ \varphi = \textbf{E} \left \langle a \right \rangle \varphi_1 $ or $ \varphi = \textbf{EG}\varphi_1 $, the corresponding LIA formula is generated through transition and path constraints under the bounded semantics. Finally, it returns a closed LIA formula. 
	\begin{algorithm}[h]
		\label{algo}
		\SetKwInOut{Input}{input}
		\SetKwInOut{Output}{output}
		\caption{ $ Trans(\varphi, s, k) $}
		\Input {EG-formula $ \varphi $, n-dimensional vector $ s=(s_1,..., s_n) $, natural number $k$}
		
		\Output {linear integer arithmetic formula $ \Psi $}
		\Begin{
			\Case{$ \varphi = \textbf{am}^T \ge \textbf{b} $}{
				$ \psi := \textbf{a}s^T \ge \textbf{b} $
			}
			\Case{$ \varphi = \lnot \varphi_1$} {
				$ \psi := \lnot Trans(\varphi_1, s, k) $
			}
			\Case{$ \varphi = \varphi_1 \wedge \varphi_2 $}{
				$ \psi := Trans(\varphi_1, s, k) \wedge Trans(\varphi_2, s, k) $
			}
			\Case{$ \varphi = E \left \langle a \right \rangle \varphi_1$}{
				$ \psi := k \ge 1 \wedge \exists t_1 \cdots \exists t_n. \wedge (T(s,(t_1,\cdots,t_n),a) \wedge Trans(\varphi_1, (t_1,\cdots, t_n), k)) $
			}
			\Case{$ \varphi = EG\varphi_1 $}{
				$ \theta_1 := Path((u(0)_1, \cdots, u(0)_n), \cdots, (u(k)_1, \cdots, u(k)_n)) $;\\
				$ \theta_2 := \bigwedge_{i=1}^n u(0)_i = s_i$;\\
				$ \theta_3 := \bigwedge_{j=0}^k Trans(\varphi_1,(u(j)_1,\cdots,u(j)_n),k)$;\\
				$ \psi := \exists u(0)_1 \cdots \exists u(0)_n \cdots \exists u(k)_1 \cdots \exists u(k)_n.\theta_1 \wedge \theta_2 \wedge \theta_3 $
			}
			\Return $ \psi $
		}
	\end{algorithm}

	We state that Algorithm~\ref{algo} is sound: If all components of the input vector $ s $ are greater than 0 and the generated formula $ \psi $ is satisfied, then any variable $x$ in $ \psi $ are non-negative, (i.e. $ x \ge 0 $). As decribed above, we give the following lemma and theorem: 
	\begin{lemma}
		\label{lemma1}
		Given $ r \in \Delta $, two $ n $-dimensional vectors: $ s= (s_1, \cdots, s_n) $ and $t = (t_1, \cdots, t_n) $. If $ s(^\bullet r) \ge 1$, $ T^{-}(s,t,r) $, and $ \forall 1 \le i \le n, s_i > 0$, then $ \forall 1\le i \le n, t_i > 0 $. 
		\end{lemma}
	\begin{proof}
		According to the definition of $ T^{-}(s,t,r)$ and mapping $ P^- $,  for any $ \forall 1 \le i \le n $, we can get that $ t_i = s_i + P^{-}(^\bullet r , r^\bullet)_i = s_i + c^{-}(^\bullet r, r^\bullet,X_i) $. At the same time, $ \forall r \in \Delta, r^\bullet(X_i) \ge 0 $ is satisfied. We can discuss the following two cases:
		\begin{itemize}
			\item If $ X_i $ does not equal to $ ^\bullet r $, then $ t_i = s_i + r^\bullet(X_i) \ge s_i > 0 $;
			\item If $ X_i $ equals to $ ^\bullet r $, then $ t_i = s_i + r^\bullet(X_i) - 1 =  r^\bullet(X_i) + s(X_i) - 1 =  r^\bullet(X_i) + s(^\bullet r) - 1  \ge r^\bullet(X_i) \ge 0$.
			\end{itemize}
		This concludes the proof of lemma~\ref{lemma1}. \qed
		\end{proof}
	
	Using the Lemma~\ref{lemma1}, we can now prove the following theorem:
	\begin{theorem}
		\label{theorem3}
		Given an \textbf{EG}-formula $ \varphi $, a $ n $-dimensional vector $ s= (s_1, \cdots, s_n) $ and non-negative integer $ k $. Suppose $ \varphi $ is $ \textbf{E} \left \langle a \right \rangle \varphi_1 $ or $ \textbf{EG}\varphi_1 $. If $ \forall 1 \le i \le n$ and $Trans(\varphi, s, k)$, $s_i \ge 0$, then any additional variables $ x $ produced in the construction of LIA through $ Trans(\varphi, s, k)$, are nonnegative (i.e. $ x \ge 0$).
		\end{theorem}
	
	\begin{proof}
		According to the type of $\varphi$, we discuss the following two cases:
		\begin{itemize}
			\item Suppose that $ \varphi $ is $ \textbf{E} \left \langle a \right \rangle \varphi_1 $. 
			
			During the construction of $ Trans(\textbf{E} \left \langle a \right \rangle \varphi_1,s, k) $, new variables are $ t_1, \cdots, t_n $, and $ T(s,(t_1,\cdots, t_n),a) $ is satisfied. Then from the definition of \\$ T(s,(t_1,\cdots, t_n),a) $, there exists an $ r\in \Delta $, s.t. $ r^{act}=a, s(^\bullet r)\ge 1 $, and $ T^{-}(s,t,r) $. So from Lemma~\ref{lemma1}, $ \forall 1 \le i \le n,  t_i \ge 0 $.
			\item Suppose that $ \varphi $ is $ \textbf{EG}\varphi_1$. 
			
			During the construction of $ Trans(\textbf{EG}\varphi_1,s, k) $, new variables are $ u(0)_1, \cdots,\\ u(0)_n, \cdots, u(k)_1, \cdots, u(k)_n $. We prove this theorem by structural induction on $ j $, the step size. 
			
			Basis: For $ j=0 $, $ \forall 1 \le i \le n, u(0)_i = s_i \ge 0 $. Therefore new variables $  u(0)_1, \cdots, u(0)_n $ are non-negative.
			
			Induction step: Assume that $ j \ge 1 $ and the theorem is true for $ j-1 $ steps, i.e. $ u(j-1)_1, \cdots, u(j-1)_n$ are non-negative. According to $ Path((u(0)_1, \cdots,\\ u(0)_n), \cdots, (u(k)_1, \cdots, u(k)_n)) $, there exists $ r \in \Delta $, s.t. $ u(j-1)(^\bullet r)  \ge 1$ and $ T^{-}(u(j-1), u(j), r) $ is satisfied. So from Lemma~\ref{lemma1}, $ u(j)_1, \cdots, u(j)_n $ are non-negative.
			\end{itemize}
		This concludes the proof of Theorem~\ref{theorem3}. \qed
		\end{proof}
	According to Algorithm~\ref{algo}, it will produce new variables only when meet $ \textbf{E} \left \langle a \right \rangle$ operators and \textbf{EG} operators. So according to Theorem~\ref{theorem3}, as long as the input $ s $ is a legal BPP state (i.e. each component of the vector $ s $ is non-negative), the coding based on Algorithm~\ref{algo} will not produce an illegal BPP state. Then the reliability and soundness theorem of Algorithm~\ref{algo} is stated as follows:
	\begin{theorem}
		\label{theorem4} Given an \textbf{EG}-formula $ \varphi $, a n-dimensional vector $ s $, and a natural number $ k $, Algorithm~\ref{algo} is correct: 
		\begin{itemize}
			\item Termination: The algorithm will terminate after a finite number of steps.
			\item Reliability: If $ s \models_k \varphi $, then $ Trans(\varphi, s, k) $ is satisfied.
			\item Soundness: If $ Trans(\varphi, s, k) $ is satisfied, then $ s \models_k \varphi $.
		\end{itemize}
	\end{theorem}

	\begin{proof}
		Since Algorithm~\ref{algo} constructs LIA formula recursively over the structure of \textbf{EG}-formula $ \varphi $, it can be concluded that Algorithm~\ref{algo} terminates. For reliability and soundness, we prove their correctness by structural induction on $ \varphi $:
	\begin{itemize}
		\item Suppose that $ \varphi $ is an atomic proposition, i.e. $ \varphi = \textbf{am}^T \ge b $, then from the definition $ s\models_k \textbf{am}^T \ge b $, reliability and soundness are satisfied.
		
		\item Suppose that $ \varphi = \lnot \varphi_1 $. 
		
		For reliability, assume that $ Trans(\lnot \varphi_1, s, k ) $ is unsatisfied, then according to Algorithm~\ref{algo}, $ \lnot Trans(\varphi_1, s, k ) $ is unsatisfied, so $ Trans(\varphi_1, s, k ) $ is correct. By induction hypothesis, $ s\models_k \varphi_1 $ is satisfied so $ s\models_k \lnot \varphi_1 $ is unsatisfied. 
		
		For soundness, whereas, assume that $ Trans(\lnot \varphi_1, s, k ) $ is correct, then according to Algorithm~\ref{algo}, $ \lnot Trans(\varphi_1, s, k ) $ is unsatisfied. By induction hypothesis, $ s\models_k \varphi_1 $ is unsatisfied so $ s\models_k \lnot \varphi_1 $ is satisfied.
		
		\item  Suppose that $ \varphi = \varphi_1 \wedge \varphi_2$. 
		
		For reliability, assume that $ Trans(\varphi = \varphi_1 \wedge \varphi_2, s, k ) $ is unsatisfied, then according to Algorithm~\ref{algo}, $ Trans(\varphi_1, s, k ) \wedge Trans(\varphi_2, s, k )$ is unsatisfied, so $ Trans(\varphi_1, s, k ) $ is unsatisfied or $ Trans(\varphi_2, s, k ) $ is unsatisfied. By induction hypothesis, $ s \models_k \varphi_1 $ or $ s\models_k \varphi_2 $ is unsatisfied so $ s\models_k \varphi_1 \wedge \varphi_2$ is unsatisfied. 
		
		For soundness, assume that $ Trans(\varphi_1 \wedge \varphi_2, s, k ) $ is correct, then according to Algorithm~\ref{algo}, $ \lnot Trans(\varphi_1, s, k ) $ is unsatisfied and $ \lnot Trans(\varphi_2, s, k ) $ is unsatisfied. By induction hypothesis, $ s\models_k \lnot \varphi_1 $ and $ s\models_k \lnot \varphi_2 $ are unsatisfied so $ s\models_k \varphi_1 \wedge \varphi_2 $ is satisfied.
		
		\item  Suppose that $ \varphi = \textbf{E} \left \langle a \right \rangle \varphi_1 $. 
		
		For reliability, assume that $ Trans(\textbf{E} \left \langle a \right \rangle \varphi_1, s, k ) $ is unsatisfied, then $ k= 0 $ or, for any $ t_1, \cdots,t_n $, if $ T(s, (t_1, \cdots,t_n), a) $, then $ \lnot Trans(\textbf{E} \left \langle a \right \rangle \varphi_1, s, k ) $. If $ k = 0 $, obviously $ s\models_k \textbf{A}\left \langle a \right \rangle \lnot \varphi_1 $. For the latter, assume that $ T(s, (t_1, \cdots,t_n), a) $ is correct for any $ t_1,\cdots, t_n $, then $ \lnot Trans(\varphi_1, (t_1, \cdots,t_n), k ) $ is correct, therefore $Trans(\varphi_1, (t_1, \cdots,t_n), k ) $ is unsatisfied. By induction hypothesis,\\ $ (t_1, \cdots,t_n) \models_k \varphi_1 $ is unsatisfied, so $ s\models_k A \left \langle a \right \rangle \lnot \varphi_1$ is satisfied. In conclusion, $ s\models_k \textbf{E} \left \langle a \right \rangle \varphi_1 $ is unsatisfied. 
		
		For soundness,  assume that $ Trans(\textbf{E} \left \langle a \right \rangle \varphi_1, s, k ) $ is correct, then $ k \ge 1 $ and there exists $ (t_1, \cdots, t_n) $, s.t. $ T(s, (t_1, \cdots,t_n), a) \wedge Trans(\varphi_1, (t_1, \cdots,t_n), k ) $, so $\lnot  Trans(\textbf{E} \left \langle a \right \rangle \varphi_1, (t_1, \cdots,t_n), k ) $ is unsatisfied. By induction hypothesis,\\$ (t_1, \cdots,t_n) \models_k \lnot \varphi_1 $ is unsatisfied, so $ (t_1, \cdots,t_n) \models_k \varphi_1$ is satisfied. In conclusion, $ s\models_k \textbf{E} \left \langle a \right \rangle \varphi_1 $ is satisfied. 
		
		\item  Suppose that $ \varphi = \textbf{EG}\varphi_1 $. Let $ \varphi_{conj} = \theta_1 \wedge \theta_2 $ in the Algorithm~\ref{algo}. 
		
		For reliability, assume that $ Trans(\textbf{EG}\varphi_1, s, k ) $ is unsatisfied, then for any $ u(0)_1, \cdots, u(0)_n, \cdots, u(k)_1, \cdots, u(k)_n $, if formula $ \varphi_{conj} $ is satisfied, then \\$ \lnot \bigwedge_{j=0}^k Trans(\varphi_1, (u(j)_1,\cdots, u(j)_n),k) $. So there exists $ 0\le j\le k $,\\ s.t. $ \lnot Trans(\varphi_1, u(j),k) $ is satisfied, so we find $  Trans(\varphi_1, u(j),k) $ unsatisfied. By induction hypothesis, $ u(j) \models_k \varphi_1 $ is unsatisfied, so $ s \models_k\textbf{ AF} \lnot \varphi_1 $ is satisfied, i.e. $ s \models_k\textbf{EG} \varphi_1 $ is unsatisfied.
		
		For soundness, assume that $ Trans(\textbf{EG}\varphi_1, s, k ) $ is correct, then there exists $ u(0)_1, \cdots, u(0)_n, \cdots, u(k)_1, \cdots, u(k)_n $, s.t. $ \varphi_{conj} \wedge \bigwedge_{j=0}^k Trans(\varphi_1, (u(j)_1,\cdots,\\ u(j)_n),k)$. So for any $ 0\le j \le k $, $ \lnot Trans(\varphi_1, u(j), k) $ is unsatisfied. By induction hypothesis, $ u(j)\models_k \lnot \varphi_1 $ is unsatisfied, i.e. $ u(j)\models_k \varphi_1 $ is satisfied. In conclusion, $s \models_k\textbf{EG} \varphi_1 $ is satisfied.
	\end{itemize}
	This concludes the proof of Theorem~\ref{theorem4}. \qed
\end{proof}
	
	\section{ACS2BPP: Reduction from ACS to BPP}
	\label{ACSBPP}
	In this section, we introduce the ACS2BPP module of BPPChecker. We give Actor Communicating System (ACS) the over-approximation BPP-based semantics to reduce ACS to BPP. With the support of the Erlang verifier Soter and ACS2BPP module, we can easily transfer Erlang programs to ACS and then verify \textbf{EF}-formulas defined safety properties on ACS. 
	
	\subsection{BPP-based Semantics of ACS} 
	Considering the case when the rule $ q_1 \stackrel{p?m}{\longrightarrow} q_2$ of ACS is used, we can observe that two symbols will decrease at the same time: (i) The number of state $ q_1 $ will decrease as it transfers to $ q_2 $; (ii) The side-effect of $ p?m $ will consume a message $ m $ from its mailbox. In this case, the value of two different symbols should be decreased at the same time, whereas the rule of BPP only allows one symbol to appear on the left side. 
	
	So we adapt two rules of ACS to BPP so that we can use BPP to simulate and verify ACS~\cite{30}. The modification is to add two labels $ in $ and $ out $ to each message, respectively recording the history of the message arriving and leaving the mailbox. The mapping from mailbox to message $ P \times M $ is now replaced with $ P \times \{M \times \{in, out\}\} $, then the BPP semantics of ACS is defined as follows: 
	
	Given an ACS $ \mathcal{A}=(Q, P, M, R) $, we construct BPP$ (V,\Delta) $, where $ V = Q \cup (P \times M \times \{in, out\})$. For $ q_1 \stackrel{op}{\longrightarrow} q_2 \in R $, according to the type of $ op $, we define $ \Delta $ as : 
	\begin{itemize}
		\item  $ op = nop $, add a rule $ q_1 \rightarrow q_2 $;
		\item $ op = \nu q_3 $, add a rule $ q_1 \rightarrow q_2 || q_3$;
		\item $ op = p!m $, add a rule $ q_1 \rightarrow q_2 || (p, m^{in}) $;
		\item $ op = p?m $, add a rule $ q_1 \rightarrow q_2 || (p, m^{out}) $
	\end{itemize}
	
	\subsection{Place Convert Function}
	We also give an algorithm called \textit{place convert function} for transforming places from original semantics to BPP-based semantics.
	
	Given an ACS $ \mathcal{A}=(Q, P, M, R) $, where $ |Q| = x \in \mathbb{N}$, input a place $ c_\mathcal{A} = (u_\mathcal{A}, v_\mathcal{A}) $ under the original semantics to \textit{place convert function} will output a new place $ c_\mathcal{B}=(u_\mathcal{B}, v_\mathcal{B}) $ satisfying transformation as follows:
	\begin{itemize}
		\item $ \forall 1 \le i \le x$, $ u_\mathcal{B}[q_i] =u_\mathcal{A}[q_i]$;
		\item $ \forall p\in P, m \in M$, $ v_\mathcal{B}[(p, m^{in})] = v_\mathcal{A}[(p, m)] $;
		\item $ \forall p \in P, m \in M $,$ v_\mathcal{B}[(p, m^{out})] = 0 $.
	\end{itemize}
	
	\indent	
	Under the new semantics, the state counter remains unchanged, while the message counter change. For process $ p $ and message $ m $, the number of $ m^{in} $ in $ p $ equals to the number of $ m $ in $ p $ in the original semantics and the initial count $ m^{out} $ in $ p $ is 0. When $ p $ consumes message $ m $, the original rule reduces $ m $ in $ p $, while the new rule increases $ m^{out} $ in p.
	
	We state that the BPP-based semantics of ACS is an over-approximation of ACS, which means it contains more behaviors than the original semantics. The following Theorem~\ref{theorem:convert} proves that for any reachable place $ c_\mathcal{A} $ under the original semantics, there exists an reachable place $ c_\mathcal{B} $, whose number of states is the same as that of $ c_\mathcal{A} $ and the difference between $ (p, m^{in}) $ and $ (p, m^{out}) $ in $ c_\mathcal{B} $ equals to the number of $ (p,m) $ in $ c_\mathcal{A} $.
	\begin{theorem}
		\label{theorem:convert}
		Given an ACS $ \mathcal{A}=(Q, P, M, R) $, where $ |Q| = x, |P| = y, |M| = z$, where $ x,y,z \in \mathbb{N}$, and given a place $ c_{\mathcal{A}_0} $ under the original semantics, if there exists $ c_\mathcal{A} = (u_\mathcal{A},v_\mathcal{A})$ that satisfies:
		\begin{itemize}
			\item $ \forall 1\le i \le x, u_\mathcal{A}[q_i] =k_i$;
			\item $ \forall 1\le i \le y, v_\mathcal{A}(p_i) = m_1^{h_{i1}}\cdots m_z^{h_{iz}}$, where $ v_\mathcal{A}(p_i) $ denotes messages of $ p_i $;
			\item $ c_{\mathcal{A}_0} \Rightarrow^* c_\mathcal{A} $,
		\end{itemize}
	then there exists a place $ c_\mathcal{B} = (u_\mathcal{B}, v_\mathcal{B})$ under the BPP-based semantics such that:
		\begin{itemize}
		\item $ \forall 1\le i \le x, u_\mathcal{B}[q_i] =k_i$;
		\item $ \forall 1\le i \le y, v_\mathcal{B}(p_i) = {m_1^{in}}^{r_{i1}} {m_1^{out}}^{s_{i1}}\cdots {m_z^{in}}^{r_{iz}} {m_z^{out}}^{s_{iz}}$, where $ v_\mathcal{B}(p_i) $ denotes messages of $ p_i $;
		\item $ \forall 1\le i \le y, \forall 1\le j \le z, r_{ij}-s_{ij} =h_{ij} $;
		\item $ c_{\mathcal{B}_0} \Rightarrow^\star c_\mathcal{B} $, where $ c_{\mathcal{B}_0} $ is the output of \textit{place convert function} with the input $ c_{\mathcal{A}_0} $.
	\end{itemize}
		\end{theorem}
	
	\begin{proof}
		We prove the correctness of Theorem~\ref{theorem:convert} by structural induction:
		
		Basis step: $ c_\mathcal{A} $ is $ c_{\mathcal{A}_0} $, then the output $ c_{\mathcal{B}_0} $ is the place of new semantics.
		
		Induction step: Without loss of generality, let $ r = q_1\stackrel{op}{\longrightarrow} q_2$, and $ c_{\mathcal{A}_0} \Rightarrow^\star c^{'}_\mathcal{A} \stackrel{r}{\Rightarrow} c_\mathcal{A}$, then we discuss the cases of $ op $ as below:
		\begin{itemize}
			\item $ \textbf{op=nop} $. 
			
			According to the rule of ACS semantics $ c^{'}_\mathcal{A}=(u^{'}_\mathcal{A},v^{'}_\mathcal{A}) $ and $ u^{'}_\mathcal{A}[q_1]=k_1+1, u^{'}_\mathcal{A}[q_2]=k_2-1 $ is satisfied. By induction hypothesis, there exists the place $ (u^{'}_\mathcal{B},v^{'}_\mathcal{B}) $ under the new semantics, s.t. $ (u_{\mathcal{B}_0}, v_{\mathcal{B}_0}) \Rightarrow^\star (u^{'}_\mathcal{B},v^{'}_\mathcal{B})$, and at the same time $ u^{'}_\mathcal{B} = u^{'}_\mathcal{A}$. So $ u^{'}_\mathcal{B}[q_1] = u^{'}_\mathcal{A}[q_1] = k_1 +1 \ge 1$. Therefore we can get $ (u^{'}_\mathcal{B},v^{'}_\mathcal{B}) \stackrel{r}{\Rightarrow} (u_\mathcal{B},v_\mathcal{B}) $, in which $ u_\mathcal{B}[q_1] = u^{'}_\mathcal{B}[q_1]-1=k_1$ and $ u_\mathcal{B}[q_2] = u^{'}_\mathcal{B}[q_2]+1=k_2$. Beyond that $ nop $ has no other side-effect, so the place $ (u_\mathcal{B},v_\mathcal{B}) $ is eligible.
			
			\item $ \textbf{op=} \nu \textbf{q}_3 $. 
			
			According to the rule of ACS semantics $ c^{'}_\mathcal{A}=(u^{'}_\mathcal{A},v^{'}_\mathcal{A}) $ and $ u^{'}_\mathcal{A}[q_1]=k_1+1, u^{'}_\mathcal{A}[q_2]=k_2-1,  u^{'}_\mathcal{A}[q_3]=k_3-1$ is satisfied. By induction hypothesis, there exists the place $(u^{'}_\mathcal{B},v^{'}_\mathcal{B}) $ under the new semantics, s.t. $ (u_{\mathcal{B}_0}, v_{\mathcal{B}_0}) \Rightarrow^\star (u^{'}_\mathcal{B},v^{'}_\mathcal{B})$, and at the same time $ u^{'}_\mathcal{B} = u^{'}_\mathcal{A}$. Therefore we can get $ (u^{'}_\mathcal{B},v^{'}_\mathcal{B}) \stackrel{r}{\Rightarrow} (u_\mathcal{B},v_\mathcal{B}) $, in which $ u_\mathcal{B}[q_1] = u^{'}_\mathcal{B}[q_1]-1=k_1$ , $ u_\mathcal{B}[q_2] = u^{'}_\mathcal{B}[q_2]+1=k_2$ and $ u_\mathcal{B}[q_3] = u^{'}_\mathcal{B}[q_3]+1=k_3$. Beyond that $ (u_\mathcal{B},v_\mathcal{B}) $ is consistent with  $ (u^{'}_\mathcal{B},v^{'}_\mathcal{B})$. So $ c_\mathcal{B}=(u_\mathcal{B},v_\mathcal{B}) $ is eligible.
			
			\item $ \textbf{op=}{p_i}!{m_j} $. 
			
			According to the rule of ACS semantics $(u^{'}_\mathcal{A},v^{'}_\mathcal{A}) $ satisfies that $ u^{'}_\mathcal{A}[q_1]=k_1+1, u^{'}_\mathcal{A}[q_2]=k_2-1,  v^{'}_\mathcal{A}[(p_i,m_j)]=h_{ij}-1 $. By induction hypothesis, there exists the place $(u^{'}_\mathcal{B},v^{'}_\mathcal{B}) $ under the new semantics, s.t. $ (u_{\mathcal{B}_0}, v_{\mathcal{B}_0}) \Rightarrow^\star (u^{'}_\mathcal{B},v^{'}_\mathcal{B})$, and at the same time $ u^{'}_\mathcal{B} = u^{'}_\mathcal{A}$, $ v^{'}_\mathcal{B}[(p_i,{m_j}^{in})]-v^{'}_\mathcal{B}[(p_i,{m_j}^{out})]=h_{ij}-1 $. 
			We can get $ (u^{'}_\mathcal{B},v^{'}_\mathcal{B}) \stackrel{r}{\Rightarrow} (u_\mathcal{B},v_\mathcal{B}) $, in which $ u_\mathcal{B}[q_1] = u^{'}_\mathcal{B}[q_1]-1=k_1$ , $ u_\mathcal{B}[q_2] = u^{'}_\mathcal{B}[q_2]+1=k_2$ and $ v_\mathcal{B}[(p_i,{m_j}^{in})] = v^{'}_\mathcal{B}[(p_i,{m_j}^{in})]+1$. So, 
			\begin{equation*}
				\begin{aligned}
					&v_\mathcal{B}[(p_i,{m_j}^{in})] - v_\mathcal{B}[(p_i,{m_j}^{out})] \\
					&= v^{'}_\mathcal{B}[(p_i,{m_j}^{in})]+1 - v^{'}_\mathcal{B}[(p_i,{m_j}^{out})] \\
					&= h_{ij}-1+1 = h_{ij} \\
					&= v_\mathcal{A}[(p_i,m_j)]
				\end{aligned}
			\end{equation*}
			Beyond that other elements are consistent, so $ (u_{\mathcal{B}_0}, v_{\mathcal{B}_0}) \Rightarrow^\star (u_\mathcal{B},v_\mathcal{B}) $ and $ (u_\mathcal{B},v_\mathcal{B}) $ is eligible.
			
			\item $ \textbf{op}=p_i?m_j $. 
			
			According to the rule of ACS semantics $(u^{'}_\mathcal{A},v^{'}_\mathcal{A}) $ satisfies that $ u^{'}_\mathcal{A}[q_1]=k_1+1, u^{'}_\mathcal{A}[q_2]=k_2-1,  v^{'}_\mathcal{A}[(p_i,m_j)]=h_{ij}+1 $. By induction hypothesis, there exists the place $(u^{'}_\mathcal{B},v^{'}_\mathcal{B}) $ under the new semantics, s.t. $ (u_{\mathcal{B}_0}, v_{\mathcal{B}_0}) \Rightarrow^\star (u^{'}_\mathcal{B},v^{'}_\mathcal{B})$, and at the same time $ u^{'}_\mathcal{B} = u^{'}_\mathcal{A}$, $ v^{'}_\mathcal{B}[(p_i,{m_j}^{in})]-v^{'}_\mathcal{B}[(p_i,{m_j}^{out})]=h_{ij}+1 $. 
			We can get $ (u^{'}_\mathcal{B},v^{'}_\mathcal{B}) \stackrel{r}{\Rightarrow} (u_\mathcal{B},v_\mathcal{B}) $, in which $ u_\mathcal{B}[q_1] = u^{'}_\mathcal{B}[q_1]-1=k_1$ , $ u_\mathcal{B}[q_2] = u^{'}_\mathcal{B}[q_2]+1=k_2$ and $ v_\mathcal{B}[(p_i,{m_j}^{out})] = v^{'}_\mathcal{B}[(p_i,{m_j}^{out})]+1$. So, 
			\begin{equation*}
				\begin{aligned}
					&v_\mathcal{B}[(p_i,{m_j}^{in})] - v_\mathcal{B}[(p_i,{m_j}^{out})] \\
					&=v_\mathcal{B}[(p_i,{m_j}^{in})]-( v^{'}_\mathcal{B}[(p_i,{m_j}^{out})]+1) \\
					&= h_{ij}+1-1 = h_{ij} \\
					&= v_\mathcal{A}[(p_i,m_j)]
				\end{aligned}
			\end{equation*}
			Beyond that other elements are consistent, so $ (u_{\mathcal{B}_0}, v_{\mathcal{B}_0}) \Rightarrow^\star (u_\mathcal{B},v_\mathcal{B}) $ and $ (u_\mathcal{B},v_\mathcal{B}) $ is eligible.
		\end{itemize}
		This concludes the proof of Theorem~\ref{theorem:convert}. \qed
		\end{proof}
	So when the reachability (\textbf{EF}-formulas) of BPP is unsatisfied, the corresponding state in ACS is also unreachable. 
	
	Based on above techniques, our tool can verify safety properties of ACS and even Erlang programs more efficiently because we reduce the complexity from EXPSPACE-complete to NP-complete~\cite{9,28}. More specifically, safety properties include the unreachability of errors, boundedness of mailbox in system etc. Here we omit formal descriptions of safety properties due to space considerations and a case study is given in the following section to help understand.
	
	\subsection{Case Study}
	We give a case study of how BPP-based model checking of ACS combining techniques proposed above works.
	
	Given an ACS $ \mathcal{A}=(Q, P, M, R) $, where $ Q = \{q_0, q_1\}$, $ P = \{p\} $, $ M = \{m\}$ and $ R = \{q_0 \stackrel{p!m}{\longrightarrow} q_1, q_1 \stackrel{p?m}{\longrightarrow} q_0\} $. We assume that the initial place $ (u, v) $ under original semantics satisfies: $ u[q_0] = 1, u[q_1] = 0$ and $v[(p,m)] = 0 $. 
	
	Figure~\ref{fig:acs_example} shows the ACS generated by our BPP-based semantics. Each node represents a place and each directed edge represents a transition rule. The superscripted item like $ (p,m^{in})^k $ means the mapping of place satisfies $ v[(p, m^{in})] == k $.
	
	\begin{figure}
		\centering
		\includegraphics[width=0.75\textwidth]{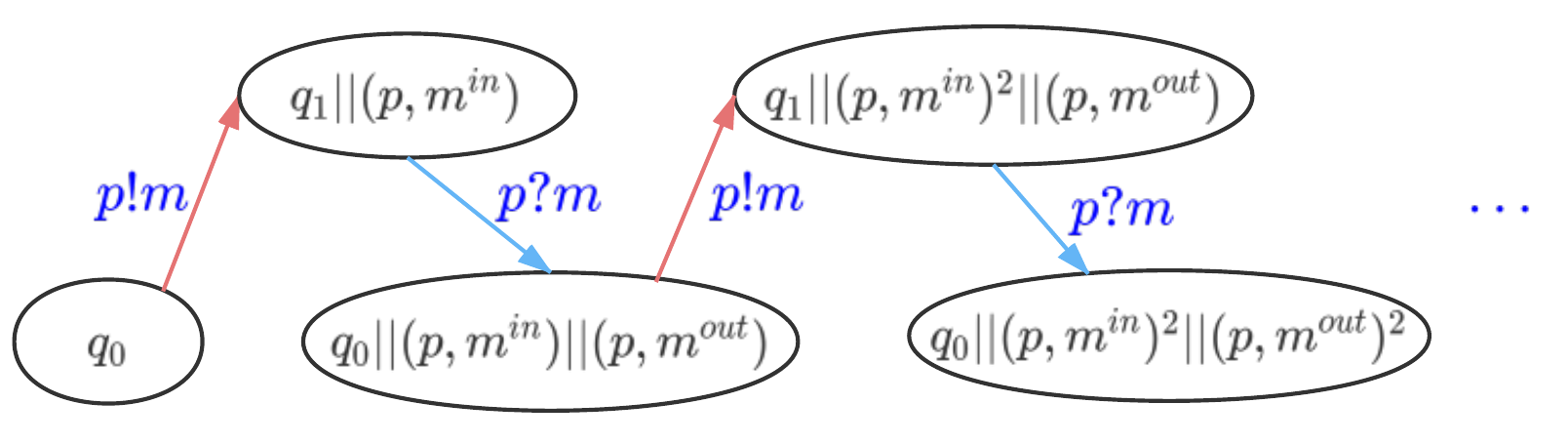}
		\caption{A BPP-based Actor Communicating System $ \mathcal{A} $}\label{fig:acs_example}
	\end{figure}
	
	It is obvious that at any time the number of $ q_0 $ and $ q_1 $ should not be larger than one, and process $ p $ contains a maximum of one message. So we can define the above two reachability properties of $ \mathcal{A} $ in the form of CTL formulas:
	(i) Neither $ \textbf{EF}(u[q_0] >= 2) $ nor $ \textbf{EF}(u[q_1] >= 2) $ can be satisfied. 
	(ii) The formula $ \textbf{EF}(v[(p, m^{in})] - v[(p, m^{out})]>= 2) $ can not be satisfied. 
	
	Given the information of BPP model and specification of properties, we can verify the safety properties of ACS by model checking the CTL on BPP and get the result.
	
	\subsubsection{Over-approximation Semantics Descriptions.} We also give an example here to show the over-approximation of our BPP-based semantics against original semantics. Consider a ACS $ \mathcal{A}=(\{q\}, \{p\}, \{m\}, R) $, where $ R $ consists of two rules: $ r_1 = q \stackrel{p!m}{\longrightarrow} q $, $ r_2 = q \stackrel{p?m}{\longrightarrow} q $. Assume that in the initial place $ c_\mathcal{A} $ mailbox of $ p $ is empty. Intuitively, the original semantics require that a message $ m $ can be consumed only when at least one message $ m $ is in the mailbox. So $ c_\mathcal{A} $ first uses $ r_1 $ to pattern-match $ m $ and then consume $ m $ using $ r_2 $. The BPP-based semantics, however, allow $ p $ to consume $ m $ without containing the $ m $, i.e., the initial place $ c_\mathcal{B} $ generated by \textit{place convert function} can use both $ r_1 $ and $ r_2 $. 
	
	\section{Tool Implementation}
	In this section, we provide more details of design and implementation of BPPChecker. The overall architecture of BPPChecker is presented in Figure~\ref{fig:BPPMC}.
	
	\subsubsection{Overview of structure.} BPPChecker uses the Python3 interface of SMT solver Z3(v4.8.5) as a library for solving linear integer arithmetic. Through input files, users provide information on BPP and CTL formulas. Through instructions, users can specify the step size of bounded \textbf{EG} model checking and whether to output additional information besides the solved result. BPPChecker is mainly divided into three modules: (i) syntax parser; (ii) BPP and formula model; (iii) a model checker which implements the algorithms in Section~\ref{sec3}.
	\begin{figure}[htbp]
		\includegraphics[width=\textwidth]{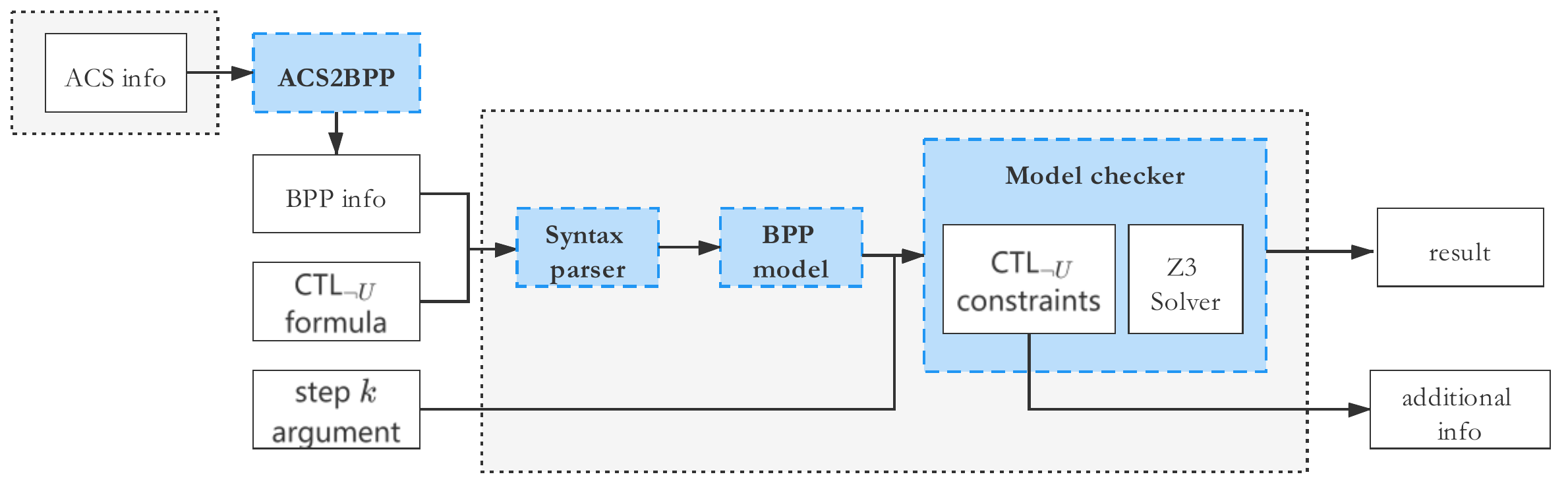}
		\caption{Architecture of BPPChecker}\label{fig:BPPMC}
	\end{figure}
	
	\subsubsection{Input and syntax parser (run.py and parser.py).} The tool first checks whether a given input file conforms to the definition of BPP and CTL grammar. We ask the user to provide information about the BPP model$ (V, \Delta) $ and the initial BPP expression. For bounded \textbf{EG} verification, the step size $k$ should be given through instruction as a parameter. The detailed syntax of input is shown in Figure~\ref{fig:grammar}. If the input is illegal, the tool reports an error and stops execution. Otherwise, the syntax parser will construct the BPP and generate the corresponding constraints to continue. Figure~\ref{input1} and Figure~\ref{input2} show an example of input to BPPChecker.
	\begin{figure}[htb]
		\centering
		\minibox[frame]{
			PROBLEM $\rightarrow$ BPP "formula" FORMULA\\
			BPP $\rightarrow$ "initial" SYMBOLS "rules" RULES \\
			SYMBOLS $\rightarrow$ VAR | SYMBOLS "," SYMBOLS\\
			RULES $\rightarrow$ RULE | RULES RULES \\
			RULE $\rightarrow$ VAR "$\rightarrow$" SYMBOLS | VAR "$\rightarrow$" LABEL "$\rightarrow$" SYMBOLS\\
			\\
			FORMULA $\rightarrow$ UNARY "(" FORMULA ")"\\
			| BINARY "(" FORMULA "," FORMULA ")"\\
			| NEXT "(" LABEL "," FORMULA ")" \\
			\\
			QUERY $\rightarrow$ ACC COMPARE NUMBER \\
			ACC $\rightarrow$ MULT | ACC CONNECT ACC\\
			MULT $\rightarrow$ VAR | VAR "$\star$" NUMBER \\
			\\
			CONNECT $\rightarrow$ [$ +- $] \\
			COMPARE $\rightarrow$ "==" | "!=" | ">=" | "<=" | ">" | "<" \\
			UNARY $\rightarrow$ "Neg" | "EG" | "AF" | "EF" \\
			BINARY $\rightarrow$ "Conj" | "Disj" | "Imp" \\
			NEXT $\rightarrow$ "EX" | "AX"\\
			\\
			VAR $\rightarrow$ [a-zA-Z][a-zA-Z0-9]$^\star$  \\
			LABEL $\rightarrow$ [a-zA-Z][a-zA-Z0-9]$^\star$ \\
			NUMBER $\rightarrow$ [1-9][0-9]$^\star$
		}
		\caption{Grammar defined inputs to BPPChecker}		\label{fig:grammar} 
	\end{figure}
	\subsubsection{ACS2BPP converter (ACS2BPP.cpp).} In the module of ACS2BPP, we implement the algorithm described in Section~\ref{ACSBPP} that converts an Actor Communicating System to BPP. With the support of ACS2BPP, given an ACS model, we can firstly input it to ACS2BPP and then verify safety properties of an ACS through BPPChecker.
	
	\subsubsection{Model checker (bpp.py, blchecker.py, rchecker.py).} The model checker of the tool can automatically verify CTL$_{\lnot U}$ on BPP. The algorithms of constraint construction are described in detail in Section~\ref{lsec} and Section~\ref{reachsec}.
	
	\subsubsection{Output.} The output contains solved result and time. Optionally, users can choose to output set of constraints constructed from statistics provided by Z3 solver. If the input problem can be satisfied, the model is returned (i.e. a set of assignment for each variable). Otherwise, constraints that caused the contradiction are returned. Figure~\ref{output} shows a satisfied model of Figure~\ref{input2}.
	\begin{figure}[htb]
		\centering
		\subfigure[]{\label{input1}
			\boxed{
				\begin{array}{l}
					initial \\
					$X$
					\\
					rules\\
					$X -> a -> Y, Z$\\
					$Y -> a -> X, Y$\\
					$Z -> b -> X$\\
					formula \\
					$\textbf{EG}(\textbf{EX}(a, Y+Z>=2))$
			\end{array}}
		}
		\subfigure[]{\label{input2}
			\boxed{
				\begin{array}{l}
					initial \\
					$S$\\
					rules\\
					$S -> X$\\
					$X -> X, Y$\\
					\\
					formula \\
					$\textbf{EF}(Y == 1)\qquad\qquad$
		\end{array}}}
		\subfigure[]{\label{output}
			\minibox[frame]{
				(x\_S, 0) \\
				(x\_X, 1) \\
				(x\_Y, 1) \\
				(y\_1, 1) \\
				(y\_2, 1) \\
				(z\_S, 0) \\
				(z\_X, 1) \\
				(z\_Y, 2)$\qquad\qquad$
			}
		}
		\caption{Input and output of BPPChecker} 
	\end{figure}
	\section{Evaluation}
	\subsection{Bounded Model Checking of EG-formula}
	In terms of BPP model, we use the BPP shown in Example~\ref{bpp} as the experimental model. Based on this BPP, we perform the model checking against the following \textbf{EG}-formulas:
	\begin{equation}
		\begin{split}	
			&\varphi_1 = \textbf{EG}(\textbf{E}\left \langle a \right \rangle(X_2 + X_3 \ge 2 ))\\
			&\varphi_2 = \textbf{EG}(X_1 + X_2 \ge 2 \rightarrow \textbf{E}\left \langle a \right \rangle(X_1 \ge 2 \wedge X_3 \ge 1 ))\\
			&\varphi_3 = \textbf{EG}(\textbf{AF}(X_1+X_2 \ge 2)	
		\end{split}
	\end{equation}
	
	Bounded model checking is conducted when step $ k $ is respectively set to 5, 10, 20 and 50. Experimental results are shown in Table~\ref{tab1}, in which we show the time(s) spent to verify each case.
	\begin{table}
		\centering
		\caption{Experiment results of bounded model checking on BPP}\label{tab1}
		\begin{tabular}{p{0.15\textwidth}p{0.1\textwidth}p{0.1\textwidth}p{0.1\textwidth}p{0.1\textwidth}}
			\toprule
			& 5 & 10 & 20 & 50\\
			\midrule
			$ \varphi_1 $ & 0.024 & 0.047& 0.072 & 0.253 \\
			$ \varphi_2 $ & 0.033 & 0.058& 0.114 & 0.245 \\
			$ \varphi_3 $ & 0.092 & 0.269& 0.981 & 6.259 \\
			\bottomrule
		\end{tabular}
	\end{table}
	
	From the experimental results, we know that the time spent of BPPChecker increases as step $ k $ increases. The nested depth of formula $ \varphi_2 $ is longer than that of $ \varphi_1 $ and $ \varphi_3$, however, the time of solving $\varphi_2 $ is only slightly longer than that of solving $ \varphi_1$, while the time of solving $ \varphi_3 $ is much longer than that of solving the other two formulas. This situation is especially obvious when the factor($ k=50 $) is large. This is because the number of variables is the main factor that limits the tool's speed. For \textbf{EG} operator, as defined in Algorithm~\ref{algo}, constructing LIA formulas needs $(k+1)\cdot n$ variables, where $n$ is the number of process symbols of BPP. Formula $\varphi_3$ is in fact logically equivalent to formula $ \textbf{EG}(\lnot \textbf{EG} \lnot(X_1 + X_2 \ge 2)) $. So the number of variables in the LIA formula generated by $ \varphi_3 $ is $ (k+1)^2\cdot n^2$. As a result, when $ k $ gets larger, the solver will spend longer time.
	
	For liveness (\textbf{EG}-formulas) properties, existing verifiers based on Petri net are oriented to traditional Petri nets, whose standard benchmarks are no longer applicable to BPP model. In order to evaluate different factors, we generate test cases randomly and run experiments on BPP with different sizes and various step sizes. Since the size of BPP is mainly reflected in the number of transition rules, we model check BPPs with transition rules of 10,20 and 30. For \textbf{EG}-formula, we still choose $ \varphi_1$, $\varphi_2$ and $ \varphi_3 $ above.
	
	\begin{table}
		\centering
		\caption{Experiment result with various numbers of transition rules and step sizes}\label{tab2}
		\begin{tabular}{p{0.2\textwidth}p{0.1\textwidth}p{0.1\textwidth}p{0.1\textwidth}p{0.1\textwidth}}
			\toprule
			\multirow{4}{2cm}{$ k = 5 $} &  & $\varphi_1$ & $\varphi_2$ & $\varphi_3$\\
			\cmidrule{2-5}
			& 10 & 0.330& 0.348 & 1.293 \\
			\cmidrule{2-5}
			& 20 & 1.153 &1.336& 3.693 \\
			\cmidrule{2-5}
			& 30 & 2.097 & 2.376 &7.920 \\
			\midrule
			\multirow{4}{2cm}{$ k = 10 $} &  & $\varphi_1$ & $\varphi_2$ & $\varphi_3$\\
			\cmidrule{2-5}
			& 10 & 0.636& 0.672 &3.785 \\
			\cmidrule{2-5}
			& 20 & 2.231 &2.587& 13.889\\
			\cmidrule{2-5}
			& 30 & 4.262 &4.706& 28.558 \\
			\midrule
			\multirow{4}{2cm}{$ k = 15 $} &  & $\varphi_1$ & $\varphi_2$ & $\varphi_3$\\
			\cmidrule{2-5}
			& 10 & 0.918& 1.050& 8.075\\
			\cmidrule{2-5}
			& 20 & 3.348& 3.809& 26.884 \\
			\cmidrule{2-5}
			& 30 & 7.249 &8.353& 56.810 \\
			\bottomrule
		\end{tabular}
	\end{table}
	
	Table~\ref{tab2} depicts experimental results in seconds(s). The time spent increases as the number of rules and $ k $ increases. The size of rule set affects the size of the disjunctive formula in the transition constraints and path constraints because we need to search the rules in BPP rule set $ \Delta $ to validate transitions and paths. The step size $ k $ mainly determines the size of conjunction formula in the path constraints of LIA formula corresponding to the outermost formula whose temporal operator is \textbf{EG}, i.e. the number of sub-constraints. Last but not least, $ k $ also has an impact on the number of variables introduced.
	
	\subsection{Model Checking on ACS}
	To evaluate the model checking on our BPP-based ACS semantics, we do experiments on ACSs generated from real Erlang asynchronously communicating programs. The benchmark we use are offered by Osualdo's work named Soter~\cite{1,10}, an automatic and efficient ACS-based model checking tool for Erlang. We experimentally compare our BPPChecker with Soter's backend BFC, where BFC verifies reachabillity (i.e. \textbf{EF}-formula) on ACS of Erlang programs and BPPChecker verifies BPPs generated by our $ ACS2BPP $ module. Table~\ref{tab3} shows that BPPChecker takes less time than the BFC in almost all test cases. Moreover, rules constructed by BPP-based semantics of ACS are also less. As a result, BPPChecker performs better in both time and number of rules. 
	
	The above experiments on real Erlang programs also show that BPPChecker has universality in real asynchronous communicating programs verfication: It supports higher-order functions such as Erlang programs, process creation, behaviors of the asynchronous messaging, verfication of error reachability problem and even mutual exclusivity (an important property of asynchronous and concurrent programs). Besides, our approach can be flexibly applied to verify number of messages in a specific mailbox, validating the effectiveness of operation order, security setting, etc.
	
	\begin{table}
		\centering
		\caption{Experiment results of verifying BPP generated by $ ACS2BPP $ }\label{tab3}
		\begin{tabular}{p{0.25\textwidth}p{0.1\textwidth}p{0.1\textwidth}p{0.1\textwidth}p{0.1\textwidth}}
			\toprule
			\multirow{2}[4]{*}{Test Case} & \multicolumn{2}{c}{BFC} & \multicolumn{2}{c}{BPPChecker}\\
			\cmidrule(lr){2-3} \cmidrule(lr){4-5} 
			& Number of Rules & Time(s) & Number of Rules & Time(s) \\
			\midrule
			pipe & 14 & 0.106& 9 & 0.025 \\
			\midrule
			ring & 40& 0.268& 28& 0.106\\
			\midrule
			state\_factory & 34& 0.736& 22& 0.074\\
			\midrule
			reslock& 52& 0.821& 38& 0.228\\
			\midrule
			reslockbeh& 62& 0.856& 46& 0.268\\
			\midrule
			parikh&31 &0.070 &20 &0.097 \\
			\bottomrule
		\end{tabular}
	\end{table}
	
	\section{Related Work}
	Verification of concurrent programs mainly relies on approximation or abstraction to limit the program model. Most verfication of concurrent programs are done on Petri nets or its extentions due to its high complexity~\cite{107,18}. 
	Sen and Viswanathan~\cite{21} proposed a multi-set pushdown system with empty stack restriction. Emmi~\cite{22} proposed an event-driven asynchronous program model, which reduced the coverage of v-PN and concluded that the state reachability of multi-set pushdown system was decidable. Kochems~\cite{24,qadeer2005context} proposed a looser K-shaped limitation to limit the number of stack symbols for receiving operations on the stack of any process and proposed a theoretical model Nets with Nested Colored Tokens (NNCT) based on Petri Nets. Osualdo proposed Actor Communicating System with finite-state processes and implemented Soter, a verifier for Erlang programs. 	 
	Despite the wide use of Petri nets~\cite{petri1966communication}, the time complexity of many branching-time properties on it is very high. Reachability problem on Petri nets has an ACKERMANN upper bound~\cite{25} and a Tower-hard lower bound~\cite{26}. Czerwińsk further improved the lower bound and proved that without restriction on dimension, the problem needs a tower of exponentials of time or space, of height exponential in input size~\cite{czerwinski2021improved}. As for coverability and boundedness, their complexity is EXPSPACE-complete~\cite{9,28}. For NNCT, coverability is TOWER-complete, while boundedness and termination are TOWER-hard~\cite{18}. So existing automatic tools such as BFC~\cite{10}, IIC~\cite{IIC} and Petrinizer~\cite{19} cannot perform well for large-scale programs, and the tools are difficult to be complete and efficient. 
	
	Different from general Petri nets, many properties on BPP have low complexity. Esparza proved the following conclusions~\cite{29,minsky1967computation}: (i) Model checking reachability on Petri nets is undecidable, but for BPP, even if finite-state BPP, the problem is PSPACE-hard; (ii) Given fixed formulas with fixed length, model checking reachability on BPP is $ \sum_d^{p} $ complete, where $ d $ is the nested depth of modal operator $ d $. On the basis of this conclusion, Mayr proved that model checking \textbf{EF} on BPP is actually PSPACE-complete~\cite{121}. 
	
	\section{Conclusion}
	Basic Parallel Process (BPP), as a subclass of Petri nets, can be used to verify concurrent programs with lower complexity. We implement BPPChecker, the first SMT-based model checker for verifying a subclass of CTL (CTL$_\lnot U$) on BPP. For EF-formulas model checking, we reduce it to the satisfiability problem of existential Presburger formula. For EG-formulas model checking, we provide a $k$-step bounded semantics and reduce it to the satisfiability problem of linear integer arithmetic. The linear integer arithmetic formulas are handled by SMT solver Z3. We give Actor Communicating System (ACS) the over-approximation BPP-based semantics and evaluate BPPChecker on ACSs generated from real Erlang programs. Experimental results show that BPPChecker performs more efficiently than the existing tools for a series of property verification problems of real Erlang programs. In the future, we plan to enhance our tool by implementing an Erlang-to-BPP converter and focus on more pracitcal problems like dynamic updating through abstract interpretation to CTL formulas on BPP.
	%
	%
	%
	\newpage
	\bibliographystyle{splncs04}
	\bibliography{mybibliography}
	
\end{document}